\documentclass[journal]{IEEEtran}
\ifCLASSINFOpdf

\else

\fi
\usepackage{algorithm} 		
\usepackage{algorithmic} 	
\usepackage{bm}
\usepackage{graphicx}
\usepackage{amsmath}
\usepackage{cite}
\usepackage{amssymb}
\usepackage{multirow}

\usepackage{epstopdf}
\usepackage{exscale}
\usepackage{relsize}
\usepackage{color}
\usepackage{graphicx}
\usepackage{subfigure}
\usepackage[table]{xcolor}
\usepackage{amssymb,amsmath}
\usepackage{cases}
\usepackage{longtable}
\usepackage{rotating}
\usepackage{multirow}
 \usepackage{siunitx}
\usepackage{epstopdf}
\usepackage{cite}
\usepackage{float}
\usepackage[marginal]{footmisc}
\usepackage{stfloats}
\usepackage[colorlinks,
            linkcolor=black,
            anchorcolor=black,
            urlcolor=black,
            citecolor=black]{hyperref}
\usepackage[hyphenbreaks]{breakurl}

\newtheorem{theorem}{Theorem}

\newtheorem{proof}{Proof}

\input epsf
\begin{document}

\title{Cost-Effective Task Offloading Scheduling for Hybrid Mobile Edge-Quantum Computing}

\author{Ziqiang Ye,\IEEEmembership{~}
Yulan Gao, \IEEEmembership{Member, IEEE,}
Yue Xiao, \IEEEmembership{Member, IEEE,}
Minrui Xu, \IEEEmembership{Student Member, IEEE,}\\
Han Yu, \IEEEmembership{Senior Member, IEEE,}
and Dusit Niyato, \IEEEmembership{Fellow, IEEE}
\thanks{Z. Ye and Y Xiao are with the National Key Laboratory of Wireless Communications, University of Electronic Science and Technology of China, Chengdu 611731, China (e-mail: yysxiaoyu@hotmail.com, xiaoyue@uestc.edu.cn).}
\thanks{Y. Gao, M. Xu, Y Han, and D. Niyato are with the School of Computer Science and Engineering, Nanyang Technological University, Singapore 639798, (e-mail: yulan.gao@ntu.edu.sg, minrui001@ntu.edu.sg, han.yu@ntu.edu.sg, dniyato@ntu.edu.sg).}}

\markboth{~~~}
{Shell \MakeLowercase{\textit{et al.}}: A Sample Article Using IEEEtran.cls for IEEE Journals}

\maketitle

\begin{abstract}
In this paper, we aim to address the challenge of hybrid mobile edge-quantum computing (MEQC) for sustainable task offloading scheduling in  mobile networks.
We develop cost-effective designs for both task offloading mode selection and resource allocation, subject to the individual link latency constraint guarantees for mobile devices, while satisfying the required success ratio for their computation tasks.
Specifically, this is a time-coupled offloading scheduling optimization problem in need of a computationally affordable and effective solution.
To this end, we propose a deep reinforcement learning (DRL)-based Lyapunov approach.
More precisely, we reformulate the original time-coupled challenge into a mixed-integer optimization problem by introducing a penalty part in terms of virtual queues constructed by time-coupled constraints to the objective function.
Subsequently, a Deep Q-Network (DQN) is adopted for task offloading mode selection.
In addition, we design the Deep Deterministic Policy Gradient (DDPG)-based algorithm for partial-task offloading decision-making.
Finally, tested in a realistic network setting,
extensive experiment results demonstrate that our proposed approach is significantly more cost-effective and sustainable compared to existing methods.
\end{abstract}


\section{Introduction}
\IEEEPARstart{T}{ask} offloading has emerged as a crucial research topic in mobile edge computing \cite{kan2018task}, particularly with the rapid advancement of mobile devices and six-generation (6G) wireless communication technologies in recent years.
The idea behind task offloading is to transfer computation-intensive tasks from resource-limited mobile devices to powerful edge servers for remote execution\cite{mahenge2022energy}, and thus improving the battery-life of mobile devices and sustainability in mobile edge network.
Mobile Edge Computing (MEC) is a promising task offloading solution that combines cloud computing and edge computing technologies to bring computation, storage and network services closer to mobile users\cite{cheng2021research,mahenge2022energy,9575181,10032267}.
There are many ongoing and completed projects related to MEC in smart cities worldwide.
For example, projects in smart cities such as {Masdar City, UAE} and {Songdo, South Korea} are adopting MEC to improve the delivery of services to citizens.
The {MEC-Health Project} \cite{mechealth}, which is founded by the European Union, aims to create a platform that utilizes MEC technology to deliver healthcare services.
By leveraging the capabilities of MEC, these projects aim to enhance the delivery of remote patient monitoring and telemedicine services, enabling real-time data analysis and reducing response times for medical interventions, ultimately improving healthcare outcomes.

Despite the promising advantages of MEC, the biggest challenge faced by network operators is meeting the soaring demand for computing and communication resources required to satisfy quality of services (QoS).
This challenge is further complicated by the ever-increasing complexity and heterogeneity of tasks, requiring a significant amount of resources to execute computation tasks with diverse complexity and urgency \cite{al2017technologies}.
Classical MEC often requires costly and energy consuming hardwares (e.g., graphics processing units (GPUs), random-access memory (RAM) and hard disk drives (HDDs).
Therefore, research on finding efficient and yet cost-effective solutions for MEC in practice is imperative.

Fortunately, quantum computing was introduced as a promising new paradigm which builds on principles of quantum mechanics. In quantum computing, each qubit can represent multiple states simultaneously through superposition, while entanglement allows for the correlation between multiple qubits
\cite{pittenger2012introduction,rietsche2022quantum}.
Precisely, the quantum superposition states can be leveraged to sample the function, with the measured outcome serving as the operation's result.
This essentially improves the microcircuit found in a conventional computer with a physical experiment.
The exceptional characteristic of quantum computing enables it to execute calculations that would require millions of years for even the fastest classical supercomputers, thereby transforming previously inconceivable computations into attainable feats\cite{rietsche2022quantum}.
Notably, quantum computing dramatically minimizes latency, achieving breakneck computing speeds that far exceed those of traditional systems. Moreover, its intrinsic security attributes, grounded in the tenets of quantum mechanics, provide an unmatched defense against cyberattacks and data breaches.
Consequently, these remarkable capabilities were spurred innovation in the MEC, seamlessly incorporating quantum computing into the existing infrastructure to meet boundless growth requirements and deliver satisfactory quality of service.

Despite the promising opportunities provided by incorporating quantum computing into mobile edge network, there remain fundamental challenges that must be addressed to enable quantum computing techniques to benefit hybrid mobile edge-quantum computing (MEQC) systems.
Pioneering contributions in this area \cite{ngoenriang2022optimal, cicconetti2022resource, ngoenriang2023dqc2o} explored the joint task offloading and heterogeneous computing resource (e.g., CPUs, GPUs, and Quantum processing Units (QPUs)) allocation for hybrid MEQC systems based on competitive performance benchmarks.
For instance, \cite{ajagekar2019quantum} explored the applications of quantum computing to energy system optimization problems. In \cite{ngoenriang2023dqc2o}, an adaptive distributed quantum computing approach--DQC2O--have been proposed to manage quantum computers and quantum networks.
In \cite{ravi2021quantum}, various trends in job execution and resources consumption/utilization on quantum cloud systems have been analyzed.
More recently, Ravi {\em et al.} \cite{ravi2021adaptive} proposed an optimized adaptive job scheduling scheme for quantum cloud environments, meticulously accounting for salient features, including queuing durations and fidelity patterns exhibited by diverse quantum machines.
Nakai {\em et al.} \cite{nakaiqubit} devised a heuristic optimization algorithm addressing the qubit allocation conundrum in the realm of distributed quantum computing, which judiciously ascertains the allocation of qubits and their associated quantum gates to the appropriate quantum processors.

Existing research on joint task offloading and resource allocation in hybrid MEQC systems are mostly designed for static scenarios.
Adaptable and efficient resource allocation and task offloading scheduling solutions for hybrid MEQC networks are required to meet the diverse and evolving needs of users, applications, and devices, such as ultra-low latency, high energy efficiency, and high scalability.
In this paper, we design a dynamic task offloading scheduling approach for minimizing the time-average weighted sum of energy consumption and time latency (WSET) within the context of a hybrid MEQC network (embedded with CPUs/GPUs and QPUs).
To tackle the time coupled problem with the objective of minimizing WSET, we propose an efficient two-stage approach capitalizing on Lyapunov optimization, DQN and DDPG algorithm.
It is worth emphasizing that we consider a practical yet demanding scenario, whereby tasks can be partially assigned across both classical and quantum computing resources of edge servers.
The novelty and contributions of this paper are as follows:
\begin{itemize}
    \item Under the hybrid MEQC paradigm, we investigate four different task offloading modes: Local Offloading Mode, Cloud Offloading to Classical Mode, Cloud Offloading to QPU Mode, and Partial Offloading to Cloud Mode. These modes are determined based on the location of mobile devices, wireless channel conditions, and the computation capability of clouds. Our goal is to formulate the SWET minimization problem, which optimizes task offloading strategies while ensuring compliance with the constraints of instantaneous time latency and success ratio of computation tasks.
    \item To address the SWET minimization problem, which also takes into account the constraint of the number of qubits required to execute the quantum circuits, we propose a DRL-based Lyapunov approach.
    This DRL-based Lyapunov approach offers provably convergent and approximate optimization solutions.
    \item We conduct numerical evaluations to assess the performance of the proposed two-stage algorithm in a realistic mobility scenario. Specifically, we adapt the Gauss-Markov Mobility Model (GMMM), originally designed for simulating mobile communication systems. By simulating the algorithm on realistic network settings, we demonstrate its superior cost-efficiency, surpassing not only a conventional MEC network but also existing baseline algorithms.
\end{itemize}

The rest of this paper is constructed as follows.
Section II provides a literature review of related work in the field.
In Section III, we introduce the system model and discuss the methods used to calculate energy consumption and latency.
Section IV formulates the optimization problem for the proposed system.
Section V presents the details of the deployment of the DRL algorithm and Lyapunov framework for solving the optimization problem.
Section VI presents the simulation results and their analysis.
Finally, in Section VII, we conclude the paper and discuss future research directions.

\section{Related Work}

\subsection{Task Offloading in MEC}
In MEC, task offloading constitutes a pivotal conundrum, requiring judicious discrimination of which tasks are to be performed locally on the mobile device and which tasks need to be offloaded to nearby edge servers for processing.
In this way, the MEC system performance in terms of energy efficiency, latency, and throughout can be improved.
In light of the pivotal role that MEC plays in the context of wireless communication and computer fields, a vast corpus of literature was been dedicated to devising task offloading methodologies that minimize costs--encompassing time latency and energy consumption--while guaranteeing to communication-computing resource limitations and stringent latency constraints.
The established theoretical tools for these offloading can be developed by the theory of {\em convex optimization} \cite{boyd2004convex}, {\em dynamic programming} \cite{chen2022dynamic},  {\em reinforcement learning} \cite{tang2020deep}, {\em genetic algorithms} \cite{li2020genetic}, and {\em particle swarm optimization} \cite{you2021efficient}.
For instance, Gao \textit{et al.} in \cite{gao2019dynamic} proposed a Lyapunov optimization-theoretic approach to social-aware task offloading and cost-effective multi-resource management in heterogeneous IoT networks.
Moreover, the authors in \cite{gao2022multi} investigated the problem of distributed multi-resource allocation (include transmit power and CPU resources) for minimizing SWET in on-device distributed federated learning systems.

\subsection{Quantum Computing}
With the advent of quantum computing, remarkable improvements in computational efficiency were realized for a multitude of computation tasks.
This cutting-edge approach to computing capitalizes on the unique characteristics of qbits, departing significantly from the conventional principles of classical computing.
In the domain of classical computing, bits are employed to signify either $0$ or $1$.
For instance, the deactivated state (i.e., off) of a switch corresponds to $0$, while the activated state (i.e., on) represents $1$.
A quantum bit, or qubit, adheres to the principles of quantum mechanics, permitting it to exist in a coherent superposition of both $0$ and $1$ states.
Consequently, a qubit can concurrently encapsulate information pertaining to both $|0\rangle$ and $|1\rangle$.
This distinctive characteristic, known as superposition \cite{steane1998quantum}, refers to the inherent property of quantum particles, thereby allowing them to exist in multiple states simultaneously until a measurement is made.
By manipulating superposed qubits, operations involving both $0$ and $1$ states can be executed concurrently.
This concept bears a resemblance to the single instruction stream multiple data stream (SIMD) parallelism found in traditional computing systems.
Crucially, the capacity to concurrently represent multiple numerical values through quantum superposition expands exponentially with the quantity of qubits.
Utilizing $N$ qubits allows for the simultaneous representation of information pertaining to $2^N$ results.
Furthermore, operations performed on these $N$ qubits effectively complete tasks on $2^N$ results concurrently.
This remarkable {\em ultra-parallel} functionality underpins the exceptional computational prowess of quantum computing.

\subsection{Mobile Edge-Quantum Computing}

Recently, the rapid expansion of data volume strained the capabilities of existing  MEC systems, which rely on classical computing technology, leading to difficulties such as latency, computational bottlenecks, data storage limitations, and energy inefficiencies in managing vast amounts of data and potentially impacting the swiftness and precision of data analysis.
This challenge is particularly pronounced for data-intensive applications such as machine learning, big data analytics, and the Metaverse \cite{mystakidis2022metaverse}.
A promising approach to addressing these challenges is the quantum computing which potentially offers superior performance in processing extensive datasets in certain computational scenarios.
Among the early contributions in this area, Ma {\em et al. } \cite{ma2022hybrid} explored the integration of quantum computing with edge networks, aiming to enhance computational power and security, and discussed the associated opportunities and challenges.
Zaman {\em et al.} \cite{ zaman2023quantum} investigated the potential of variational quantum computing and quantum machine learning for 6G ultra-reliable and low-latency communication (URLLC) by exploring the integration of quantum machine intelligence into edge networks.
Inspired by these frameworks integrating quantum computing and edge networks, a vast corpus of literature has focused on developing resource allocation techniques to enhance system performance.
For instance, Xu {\em et al.} \cite{xu2022learning} introduced a novel paradigm of MEQC, which integrates quantum capabilities into mobile edge networks, thereby reducing latency and enabling faster data processing for mobile devices.
Based on the proposed MEQC system model, Xu {\em et al.} investigated the problem of classical and quantum computation offloading in multi-user multi-server systems.
Cicconetti {\em et al.} \cite{cicconetti2022resource} designed a resource allocation strategy focusing on evaluating the necessity of quantum internets for distributed quantum computing.
Wang {\em et al.} \cite{wang2021resource} proposed a quantum-inspired reinforcement learning approach to dynamically select optimal network access and offloading strategies between edge servers via WiFi and cloud servers via 5G in vehicle networks.

In summary, existing literature mainly focuses on analysing the instantaneous performance of MEQC systems, and none of these papers study the devices' long-term task offloading scheduling accounting for user mobility and task variability.
To our knowledge, this is the first to analyze  long-term task offloading scheduling of mobile devices in hybrid MEQC system,  and propose a DRL-based Lyapunov approach to achieve cost-effective and sustainable task offloading.
\begin{figure}[!t]
\begin{minipage}[t]{0.48\textwidth}
\centering
\includegraphics[width=\textwidth]{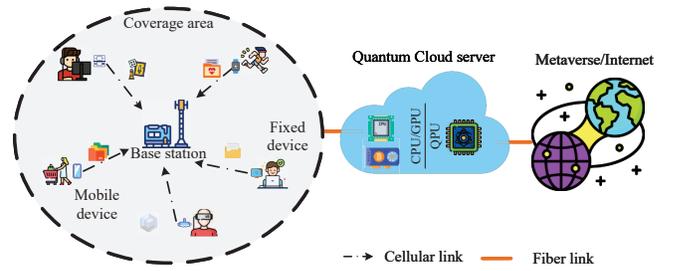}
\caption{System model.}
\label{fig:1}
\end{minipage}
\end{figure}
\section{System Model}
In this section, we introduce the system model for the proposed hybrid MEQC system.
Subsequently, we delineate four task offloading modes, each catering to different levels of available computing resources on mobile devices or edge server.

\subsection{System Overview}
As shown in Fig. \ref{fig:1}, we consider a hybrid MEQC system where $M$ heterogeneous mobile devices are adjoined to a single cloud server equipped with multiple CPUs, GPUs, and QPUs.
For notational compactness,  we will reverse the index $0$ for the central cloud.
Consider ${\cal M}:=\{1, 2, \ldots, M\}$ as the set of devices, where each mobile device is subject to dynamic computation tasks stemming from its respective application, and can randomly move around within the coverage of the corresponding access point (AP), where the AP can communicate with the central cloud via fiber links.
The introduced hybrid MEQC (i.e., quantum cloud server) depicted in Fig. \ref{fig:1} is both rational and easily deployable, as the integration quantum computing capabilities within the current MEC infrastructure proves operationally viable.
This setting can be perceived as an instantiation of the hybrid computing system outlined in  \cite{chamberlain2008visions}.
Furthermore, for the purpose of simplification in the subsequent analysis, we assume the use of a multi-access protocol (e.g., OFDMA) to establish wireless communication links for IoT devices.
Without loss of generality, we assume that each device $m\in{\cal M}$ has assigned a licensed band with bandwidth denoted by $B$.

\subsection{Computation Task Offloading Modes Design}
Our primary emphasis is directed toward the task offloading phase, during which the AP determines the task offloading mode of mobile IoT devices.
To begin with, we present the model for task offloading mode selection used in our hybrid MEQC system, where each device act based on its achieved {\em performance} (i.e., WSET).

Before designing the task offloading mode, it is essential to come to a consensus on what constitutes a computation task.
Following \cite{ngoenriang2022optimal}, the computation task of user $m$ at time slot $r$ can be denoted as a tutle $\mathcal{J}_m[r] \triangleq \{d_m[r], q_m[r], t_m^{\max}[r]\}$.
The tutle contains the data size of the task $d_m[r]$, the required CPU cycles $q_m[r]$ to execute computation task per datasize and the maximum delay $t_m^{\max}[r]$.

To successfully process the above-mentioned computation tasks within the stipulated deadlines, several challenges persist:
\begin{itemize}
    \item {Given the large number of heterogeneous devices in the network and the diversity of applications that these smart devices support, computation tasks significantly vary across devices. Ensuring each smart device can complete its respective computations within the timeline poses a significant challenge.}
    \item {The mobility of smart devices, resulting in frequent changes in their locations, leads to constant variations in the communication cost between the device and the AP.
    Ensuring timely computation task completion, even under conditions of low communication quality.}
\end{itemize}
In response to these challenges, we develop a task offloading mode that incorporates hybrid classical and quantum computing capabilities, enabling offloading to an quantum cloud server, local classical computation, and offloading to an edge servers.
To ensure our model's relevance to real-world scenarios, we take partial task offloading into account.
We introduce a continuous variable, $\phi_m\in[0,1]$,  representing the percentage of offloading to a remote server for computation by mobile device $m$.
Consequently, the computation task $(1-\phi_m)d_m[r]$ is processed locally, while $\phi_md_m[r]$ is offloaded to the cloud.
Correspondingly, the task offloading modes depicted in Fig. \ref{fig:4} have been designed based on the hybrid MEQC system's state, wherein the associated communication and computation costs are derived through mathematical expressions, respectively.
Then, we provide a comprehensive exposition of each task offloading mode.

{\bf\em 1) Local Offloading Mode: }
As depicted in Fig. \ref{fig:4}(a), when the computational task is small in size, there is no need for the mobile device to offload it to a remote server. By performing the computation locally, the device can save on transmission latency and reduce energy consumption.

As mentioned earlier, only data of size $(1-\phi_m)d_m[r]$ is processed locally.
Therefore, in the local offloading mode, a mobile device $m \in \mathcal{M}$ executes its computation task $\mathcal{J}_m[r]=\left\{(1-\phi_m) d_m[r], q_m[r], t_m^{max}[r]\right\}$.
Given the prevalent heterogeneity in device CPU performance due to varying performance requirements, we utilize the CPU frequency $f_m$ to characterize the performance of each device.
Building upon the findings of a previous study, the energy consumption of device $m$ while processing the computation task $\mathcal{J}_m[r]$ can be determined by
\begin{equation}\label{eq:1}
    ec_{m,m}[r]=\rho f_m^\zeta t_{m,m}[r],\forall m \in \mathcal{M},
\end{equation}
where $t_{m,m}[r]=\frac{(1-\phi_m[r])d_m[r]q_m^C[r]}{f_m}$ is the time spent in calculation of device $m$, $\rho$ is a constant which hinges on the average switched capacitance and the average activity factor, and  $\zeta\geq 2$ also is a constant.
\begin{figure}[!t]
	\centering
	\includegraphics[width=3.5in]{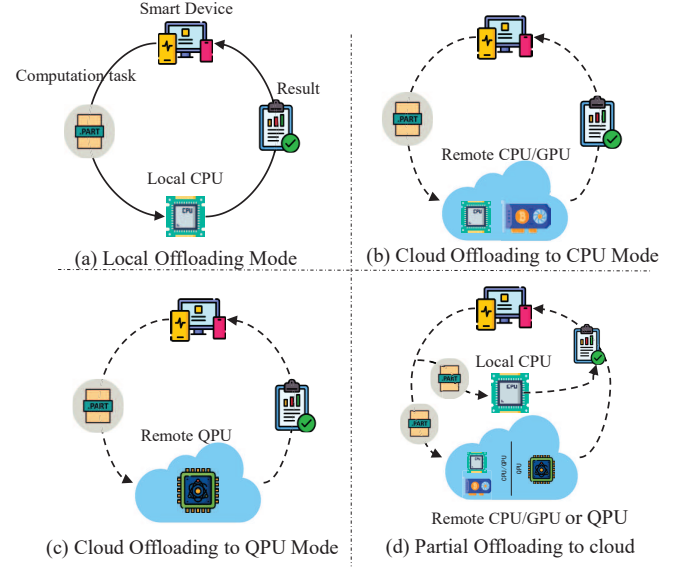}
	\caption{Illustration of the proposed four computation offloading modes in the hybrid MEQC system.}
	\label{fig:4}
\end{figure}

{\bf\em 2) Cloud Offloading to Classical Mode: }
To utilize the cloud offloading mode, as illustrated in Fig. \ref{fig:4}(b) (referred to as the classical mode), each mobile device $m$ offload the rest of data $\phi_md_m[r]$ to the AP via the cellular network. Subsequently, the AP facilitates the transmission of tasks to the remote server through optical fiber, as mentioned earlier in our discussion.

In this cloud offloading mode, depicted in Fig. \ref{fig:4}(b), the mobile device selects to utilize the CPU of the remote server for computation. The rationale behind this decision may stem from several factors. One reason is the challenge in guaranteeing the accuracy of quantum operation outcomes when employing a QPU. Another consideration could be the limited availability of quantum computing resources. Alternatively, there may be other factors that indicate the CPU is a more suitable choice for the given task.

To accurately estimate the transmission time and energy consumption in stages i) and iii),it is considered that the channel gain from mobile device $m$ to the AP is subject to path loss, shadowing, and fast-fading effects, leading to an instantaneous channel gain. Further, it is postulated that the channel gain remains constant within each time slot, being independent and identically distributed across all the slots. We denote the channel gain from device $m$ to the AP at time slot $r$ as $g_{m,0}[r]$, and $p_{m,0}[r]$ represents the corresponding transmit power.
During stage iii), the calculated result packet ensures the user is informed of the successful task execution. The size of the returned result is denoted by $d_m^{\text{ack}}$, which carries the output of the computational task.
For the sake of analytical simplicity, it is assumed that the AP transmits the results packet back to the mobile device at an constant rate.

In stage ii), the assessment method for time and energy consumption parallels that of the local offloading mode, with the exception that the CPU's performance in the cloud is superior, denoting a higher frequency. The transmission data rate achieved by mobile IoT device $m$ at time slot $r$ can then be articulated as follows:
\begin{equation}\label{eq:2}
R_{m,0}[r]=B\log_2{\left(1+\frac{p_{m,0}[r]g_{m,0}[r]}{\sigma^2}\right)},
\end{equation}
where $\sigma$ is the Additive Gaussian White Noise.

Upon introducing notations $t_{m,0}^{C,1}[r], t_{m,0}^{C,2}[r], t_{m,0}^{C,3}[r]$, the total time cost of the three stages under the mode of Fig. \ref{fig:4}(b) can be shown as:
\begin{align}\label{eq:3}
t_{m,0}^C[r]&=t_{m,0}^{C,1}[r]+t_{m,0}^{C,2}[r]+t_{m,0}^{C,3}[r],
\end{align}
where $t_{m,0}^{C,1}[r]\triangleq\frac{\phi_m[r]d_m[r]}{R_{m,0}[r]}, t_{m,0}^{C,2}[r]\triangleq \frac{\phi_m[r]d_m[r]q_m^C[r]}{f_{m,0}}$, and $t_{m,0}^{C,3}[r]\triangleq \frac{d_m^{ack}}{R_{m,0}[r]}$ represent the time cost of stage i) ii), and iii), respectively.
The CPU frequency $f_{m,0}$, allocated by the cloud to mobile device $m$, is determined by the financial contribution made by the device to the central cloud.
In addition, we assume that each device is assigned a constant CPU frequency, whereas the cloud retains the flexibility to allocate various CPU frequencies to individual devices.
The energy consumption that the server uses CPUs/GPUs to execute mobile device $m$'s computation task is
	\begin{align}\label{eq:4}
		ec_{m,0}^C[r]&=ec_{m,0}^{C,1}[r]+ec_{m,0}^{C,2}[r]+ec_{m,0}^{C,3}[r]\\
		&=p_{m,0}t_{m,0}^{C,1}[r]+\rho f_{m,0}^\zeta t_{m,0}^{C,2}[r]+p_{m,0}t_{m,0}^{C,3}[r].\nonumber
	\end{align}

{\bf\em 3) Cloud Offloading to QPU Mode: }
Executing computation tasks with QPUs can be advantageous when computation tasks are complex, and the classical computation resources of the remote server faces challenges in meeting the time deadlines (i.e., the latency requirement) for accomplishing the tasks.
Offloading computation tasks to a remote server with QPUs can drastically reduce execution time and device costs, provided users have secured sufficient quantum resources and the quantum computing error rate falls within an acceptable ratio.
Prior to formulating the energy consumption and time delay for the Cloud Offloading to QPU mode, it is important to first outline the primary steps involved in processing a computation task via quantum computing.

A quantum computer, housed within a server, operates on numerous physical qubits in an extremely low-temperature environment, approximating absolute zero. The quantum server can expedite computation tasks by compiling them into quantum circuits according to the quantum algorithms. As a result, the quantum computation task derived from $\mathcal{J}_m[r]$ can be represented by $\mathcal{J}_m[r]=\left\{d_m[r], o_m[r],l_m[r], t_m^{\max}[r]\right\}$, where $o_m[r]$ denotes the number of qubits required to execute the quantum circuit, and $l_m[r]$ denotes the depth of the quantum circuit.
However, due to the delicate nature of qubits--which can maintain a quantum state for only a brief period--contemporary quantum computers encounter difficulties in producing reliable results. This predicament suggests that the information encoded in the quantum system may be lost before the quantum computer meassure its computations. Consequently, the most formidable challenge confronting quantum computing technology is the correction of errors that occur during computation.
Quantum error correction (QEC) strategies are essential in actualizing fault-tolerant quantum computing, as they mitigate the influence of quantum noise on stored quantum information, imprecise quantum gates, faulty quantum preparation, and incorrect measurements. Implementing QEC schemes requires an array of gate operations. Given that the unique algorithms at the logical level ultimately convert into gate operations, the energy consumption associated with the error correction framework remains independent of the comprehensive algorithm design.

Our analysis focuses not on the design of top-level algorithms, but rather on the effect of QEC on energy consumption.
Let $N_1$, $N_2$, and $N_M$ defined as the average number of one-qubit (1qb) gates, two-qubit (2qb) gates, and measurement gates, respectively, that operate in parallel in each time step.
Likewise, the modes depicted in Fig. \ref{fig:4}(c)-\ref{fig:4}(d) take into account the energy consumption and time cost caused by transmission and computation.
In practice, the time cost in quantum computing mainly caused by multi-gate operations in quantum circuits\cite{fellous2022optimizing} can be obtained as:
	\begin{align}\label{eq:5}
	    t_{m,0}^Q[r]=&t_{m,0}^{Q,1}[r]+t_{m,0}^{Q,2}[r]+t_{m,0}^{Q,3}[r],
	\end{align}
where $t_{m,0}^{Q,1}[r]=\phi_m[r]d_m[r]o_m[r][\tau_1N_1+\tau_2N_2+\tau_MN_M],$ $ t_{m,0}^{Q,2}[r]=\frac{\phi_m[r]d_m[r]}{R_{m,0}[r]}$, and $t_{m,0}^{Q,3}[r]=\frac{d_m^{\text{ack}}}{R_{m,0}[r]}$ also represent the time cost of stage i) ii), and iii), respectively.
$\tau_1$, $\tau_2$ and $\tau_M$ are the processing time of 1qb gates, 2qb gates, and measurement gates, respectively.
Moreover, the detail transmission time cost in the mode of Figs. \ref{fig:4}(c) and \ref{fig:4}(d) can be derived from  \eqref{eq:2}.

As previously discussed, qubits display acute sensitivity to their surroundings.
To ensure optimal performance, quantum computers necessitate operation at temperatures approaching absolute zero.
This stipulation arises from the qubits' vulnerability to thermal vibrations, however slight, which can disrupt their state, leading to computational inaccuracies.
Consistent with the research of Fellous-Asiani {\em et al.} \cite{fellous2022optimizing}, the principal consumer of energy in quantum computing is the cooling system.
This system is integral in maintaining an environment temperature near absolute zero, a prerequisite for curtailing thermally induced vibrations that could interfere with the qubits' stability.
To approximate the energy consumption, we may employ the following  mathematical formulation.
\begin{align}\label{eq:6}
ec_{m,0}^Q[r]=&ec_{m,0}^{Q,1}[r]+ec_{m,0}^{Q,2}[r]+ec_{m,0}^{Q,3}[r]\notag\\
=&\phi_m[r]d_m[r]o_m[r](P_1N_1+P_2N_2+P_MN_M\notag\\
&+P_QQ)+p_{m,0}t_{m,0}^{Q,1}[r]+p_{m,0}t_{m,0}^{Q,3}[r].
\end{align}
Eq. \eqref{eq:6} employed for computing energy consumption encompasses the variables $P_1$, $P_2$, $P_M$, and $P_Q$.
Respectively, these represent the energy consumption associated with individual 1-qubit gates, 2-qubit gates, measurement gates, and qubits themselves.
The variable $Q$ symbolizes the count of physical qubits incorporated within a single logical qubit.
Additionally, the energy consumption for communication aligns closely with the pattern observed in cloud offloading to CPU mode.
In this model, the energy expenditure for data transfer comprises both the energy necessary for uploading computation tasks and that required for downloading values returned from remote servers.

While ultra-low temperatures and error correction mechanisms can impart a degree of reliability to quantum computations, the innate noise associated with quantum circuits persists, as highlighted by Krinner {\em et al.} in \cite{krinner2022realizing}. This noise is an intrinsic element that cannot be entirely eradicated.
Consequently, correctness often serves as a benchmark in assessing the efficacy of quantum circuits, especially with respect to non-deterministic quantum operations.
However, as the quantum circuits' scale expands, the potential error sites within the circuit multiply correspondingly.
This escalation, in turn, contributes to a decline in the computation task processing precision of the server.
Incorporating the element of noise along with the error probability for each physical gate, symbolized as $\epsilon_{err}$, as elucidated in the work of Ngoenriang {\em et al.} \cite{ngoenriang2022optimal}, enables us to arrive at a linear approximation of the success probability.
It is feasible to extract a linear approximation pertaining to the success probability, which is detailed in the following.
\begin{equation}\label{eq:7}
\mathcal{S}_m[r]=1-\mathcal{N}_m^L[r]\epsilon_{err}\left({\epsilon_{err}}/{\epsilon_{thr}}\right)^{2^{k_0}},
\end{equation}
where $\mathcal{N}_m^L[r]=o_m[r]\times l_m[r]$ represents the number of locations where logical errors may occur, $\epsilon_{trr}$ is the threshold for error correction, and $k_0$ is the level of error-correcting connectivity of the quantum server.
	
We assume that the quantum computing capacity of the quantum server leased by mobile device $m$ is $o_m^C$. Quantum computing is available only when the number of quantum circuit qubits required by the mobile device $m$ to be computed task is less than or equal to $o_m^C$ and the probability of correct operation is greater than or equal to $2/3$, which is a classical selection for a single execution of quantum algorithms\cite{fellous2022optimizing}.

\section{Problem Formulation}
In this section, the time average cost of the hybrid MEQC system will be given based on the above offloading modes.
Before giving the precise description of the cost, we first introduce the selection strategy vector ${\pmb\phi}\in {\mathbb R}^{1\times M}$ with entities $\phi_m$ representing the fraction of tasks offloaded from $m$ to the hybrid MEQC servers, where
\begin{align}\label{eq:8}
\mathcal{\phi}_{m}[r]=\begin{cases}
1, &m\text{~chooses to execute task remotely}, \\
(0,1), &m\text{~executes task partially on the server},\\
0, &m\text{~chooses to execute task locally}.
\end{cases}
\end{align}
Additionally, we introduce a binary variable, denoted as ${\cal Y}_m[r]$, to indicate whether device $m$ chooses to offload to the CPU within a hybrid MEQC server, provided $\phi_m[r]=1$ holds true.
In this scenario, we establish that
\begin{equation}\label{eq:9}
\mathcal{Y}_{m}[r]=\begin{cases}
1, & m \text{ chooses to further offload to the CPU}, \\
0, & \text{ otherwise.}
\end{cases}
\end{equation}
In an attempt to incentivize mobile devices to delegate an increased number of tasks to the cloud, we establish a long-term constraint on the continuous variable $\pmb\phi$.
Aligning with the prerequisites of low latency and quick response, we stipulate that the time-averaged cloud offloading rate for any given mobile device $m$ should not fall below a certain threshold, denoted as $\Delta\in[0,1]$.
An increase in the value of $\Delta$ corresponds to a rise in the time average of $\pmb\phi$, which, in turn, indicates a more robust inclination on the part of the user to offload computation tasks to the cloud server.
In an effort to boost the utilization rate of the remote server, we establish the condition where
\begin{equation}\label{eq:99}
    {\varlimsup_{\cal R \to \infty} } \frac{1}{\cal R}\sum\nolimits_{r=1}^{\cal R}\mathbb{E}\left\{\phi_m[r] \right\}\ge \Delta,\forall m \in \mathcal{M}.
\end{equation}

In summary, the total offloading and execution cost can be encapsulated in Eq. \eqref{eq:10}, shown at the top of the next page.
\begin{figure*}
\begin{align}\label{eq:10}
    C(\pmb{\phi}[r],\mathcal{Y}[r])=&\lambda_m^t\left((1-\mathcal{\phi}_m[r])t_{m,m}[r]+\mathcal{\phi}_m[r]\left(\mathcal{Y}_{m}[r]t_{m,0}^C[r]+(1-\mathcal{Y}_{m}[r])t_{m,0}^Q[r]\right)\right)\notag\\
    &+\lambda_m^{ec}\left((1-\mathcal{\phi}_m[r])ec_{m,m}[r]+\mathcal{\phi}_m[r]\left(\mathcal{Y}_{m}[r]ec_{m,0}^C[r]+(1-\mathcal{Y}_{m}[r])ec_{m,0}^Q[r]\right)\right).
\end{align}
\hrulefill
\end{figure*}
In Eq. \eqref{eq:10}, $\lambda_m^t$ and $\lambda_m^{ec} \in [0,1]$ are the weight parameters of latency and energy consumption for mobile device $m$.
Upon introducing Eq.(\ref{eq:10}), the problem of minimizing the total cost of task offloading in the hybrid MEQC system can be formulated as
\begin{align}
    \min_{\pmb{\phi}[r],\mathcal{Y}[r]}& C(\pmb{\phi}[r],\mathcal{Y}[r])\label{eq:11}\\
\text{s.t.~} &\mathcal{\phi}_m[r]\le 1 ~ ,\forall m \in \mathcal{M},\tag{\ref{eq:11}a}\\
    &t_m[r] \le t_m^{\max}[r] ~ ,\forall m \in \mathcal{M},\tag{\ref{eq:11}b}\\
    &\mathcal{S}_m[r]\ge 2/3 ~ ,\forall \mathcal{Y}_m[r]=0, \tag{\ref{eq:11}c}\\
    &o_m^C\ge o_m[r] ~ ,\forall \mathcal{Y}_m[r]=0\tag{\ref{eq:11}d},\\
    &\text{~Eq.~}\eqref{eq:99}.\tag{\ref{eq:11}e}
\end{align}
Constraint (\ref{eq:11}a) dictates that any mobile device $m$ must choose one of the available modes for executing a computation task, either using local computing resources or cloud resources.
Constraint (\ref{eq:11}b) ensures that mobile device $m$ will not select a mode that fails to meet the latency requirements of the computing task, as exceeding the latency threshold would result in a loss of real-time responsiveness and potentially lead to invalid results.
Concludingly, constraints (\ref{eq:11}c) and (\ref{eq:11}d) corroborate the possibility of task offloading to the QPU at the server for quantum computing, given that mobile device $m$ possesses adequate resources and the success probability of quantum computation aligns with the standards of traditional quantum algorithms.

The objective of the task offloading problem is to determine a sequence of decisions that results in optimal performance of mobile devices in terms of time average.
These problems are dynamic optimization problems and can be formulated and solved based on Lyapunov optimization algorithm.
Specifically, to transform the Eq. \eqref{eq:10} into a solvable problem, a virtual queue is introduced as $z_m[r]$ to reformulate Eq. \eqref{eq:10}.
The queues have initial condition $z_m[0]=0$ and the dynamic queue $z_m$ for mobile device $m$ is expressed as
\begin{align}\label{eq:12-1}
z_m[r+1]=\max\left\{0,z_m[r]+\Delta{\boldsymbol 1}_{\left[\phi_m[r]=0\right]}-\phi_m[r]\right\},
\end{align}
where ${\boldsymbol 1}_{\left[\phi_m[r]=0\right]}$ is an indicator function. The value becomes $1$ iff $\phi_m[r]=0$, otherwise, it evaluates to $0$.

Taking into account the quadratic Lyapunov function defined as $L[r]\triangleq \frac{1}{2} \sum_{m=1}^{M}z_m[r]^2$, we find that the single-frame drift of the Lyapunov function, denoted as $\triangle L[r]=L[r+1]-L[r]$, complies with the condition
\begin{align}
\triangle L[r] \le& \frac{1}{2}\sum\nolimits_{m=1}^{M}(\Delta {\boldsymbol 1}_{[\phi_m[r]=0]}-\phi_m[r])^2\\
&+z_m[r]\sum\nolimits_{m=1}^{M}(\Delta {\boldsymbol 1}_{[\phi_m[r]=0]}-\phi_m[r]).\notag
\end{align}
For any feasible $\mathcal{Z}[r]=\left\{z_1[r],z_2[r],\dots,z_m[r]\right\}$, taking condition expectations, the one time slot drift $\triangle L[r]$ is bounded by
\begin{align} \label{eq:15}
    \mathbb{E}\left\{\triangle L[r]|\mathcal{Z}[r]\right\} \le&
    \underbrace{\frac{Mg}{2}+\frac{d}{2}\sum\nolimits_{m=1}^{M}(\Delta)^2}_{\mathcal{B}}\\+&\mathbb{E}\left\{z_m[r]\sum_{m=1}^{M}(\Delta 1_{[\phi_m[r]=0]}-\phi_m[r])|\mathcal{Z}[r]\right\},\notag
\end{align}
where $\mathcal{B}$ is the constant that satisfies for all $r$ and all feasible $\mathcal{Z}[r]$.
Thus, our objective is to minimize the following equation with consideration of a penalty term $VC(\mathcal{\phi}[r],\mathcal{Y}[r])$
\begin{equation}\label{eq:14}
    \min \mathbb{E}\left\{-VC[r]+\sum_{m=1}^Mz_m[r](\Delta{\boldsymbol 1}_{[\phi_m[r]=0]}-\phi_m[r])|\mathcal{Z}[r]\right\},
\end{equation}
where $V>0$ is the balance parameter of the Lyapunov optimization techniques.

It is noted that directly addressing (\ref{eq:14}) poses a considerable difficulty, which is primarily due to the disparate initial tasks $\mathcal{J}_m[r]$.
Inspired by \cite{neely2012dynamic}, we can find $C$-additive approximation of the action strategy such that for a given non-negative constant $C$, we have
\begin{align}
\mathbb{E}&\left\{-VC[r]+\sum_{m=1}^Mz_m[r](\Delta {\boldsymbol 1}_{[\phi_m[r]=0]}-\phi_m[r])|\mathcal{Z}[r]\right\}\notag \\
&\leq C+\inf_{\pmb{\phi}[r],\mathcal{Y}[r]} \left\{\mathbb{E}\left\{
-VC[r]+\sum\nolimits_{m=1}^Mz_m[r]\right.\right.\notag\\
&\left.\left.\times(\Delta {\boldsymbol 1}_{[\phi_m[r]=0]}-\phi_m[r])|\mathcal{Z}[r]
\right\}
\right\}.
\end{align}
Assume the above algorithm is implemented using any $C$-additive approximation on every time slot $r$ for the minimization in (\ref{eq:14}).
Subsequently, we are able to select our variables $\pmb{\phi}[r],\mathcal{Y}[r]$ to address the optimization problem.
Then the optimization problem can be formulated as
\begin{equation}\label{eq:add}
    \min -VC[r]+\sum\nolimits_{m=1}^Mz_m[r](\Delta {\boldsymbol 1}_{[\phi_m[r]=0]}-\phi_m[r]).
\end{equation}

The following theorem establishes the tightness  of the proposed solution using the $C$-additive approximation to the optimal solution.

\begin{theorem}\label{theorem:1}
Considering the constraints of problem \eqref{eq:11} are feasible and constants $C\ge 0$, $V \ge 0$, then, for any integers $r>0$ we have
\begin{equation}
    \lim_{r \to \infty} \inf C[r] \ge C_{opt}-{(\mathcal{B}+C)}/{V}.
\end{equation}
\end{theorem}
\begin{proof}
See Appendix A.
\end{proof}

As mentioned before, the time-average constraints \eqref{eq:99} can be guaranteed if the virtual queue $\mathcal{Z}[r]=\left\{ z_1[r],z_2[r],\dots,z_M[r]\right\}$ is stable. For the {\em drift-plus-penalty} algorithm, to describe the convergence time to the desired constraints, we introduce the following theorem.
\begin{theorem}\label{theorem:2}
    For any integer $\tau >0$, the convergence time to the desired constraints within $\mathcal{O}\left( \frac{V}{\tau}\right)$, if all the conditions in Theorem \ref{theorem:1} are satisfied.
\end{theorem}

\begin{proof}
    See Appendix B.
\end{proof}

\section{Algorithm Design}
In an effort to determine the sub-optimal offloading strategy of \eqref{eq:11}, we apply the DQN and DDPG algorithms to resolve problem \eqref{eq:add} within each time slot. Subsequently, we put forth a DRL-based Lyapunov algorithm to address the optimization problem \eqref{eq:11}.
This suggested algorithm encompasses two components: an outer loop based on Lyapunov optimization and an inner loop driven by DRL. In this section, we propose a fully cooperative, multi-agent, discrete-continuous DRL algorithm designed to learn the optimal task offloading strategy.
For the successful implementation of the DQN and DDPG algorithms to solve both discrete and continuous problems, it is crucial to define the state space, action space, and reward for each agent.

\subsection{State Space}
In DQN and DDPG, the state space represents the comprehensive set of potential states that agents--in this case, the mobile devices within our hybrid MEQC system--could occupy at any given time slot within the environment.
Moreover, in the state space, we define by $\text{Observe}_m[r]=\left\{\text{Sever}_m, \mathcal{J}_m[r], g_{m,0}[r], f_m\right\}$ as the observation space of mobile device $m$ at time slot $r$. $\text{Sever}_m=\left\{f_{m,0},q_m^C\right\}$ represents the server resources leased by device $m$.

\subsection{Action Space}
In the frameworks of RL, the action space designates the set of all potential actions that an agent can perform in a given state within the environment. The action space can be either continuous or discrete. In our proposed system, we advocate for a continuous action space, which provides the agent with the flexibility to choose from a spectrum of possible actions, specifically, the fraction of tasks to offload to the remote server.
Each mobile device, denoted as $m$, maintains the option to offload its computing tasks locally, to a cloud-based CPU, or to a remote QPU. This approach of thinking step-by-step facilitates a more coherent and consistent analysis.
The computation offloading matrix can be treated as a mixed action space, combining discrete and continuous actions, which can be denoted as ${\cal A}=\{{\cal A}_m\}_{m=1}^M$, where  ${\cal A}_m[r]=\{{\phi}_m, \mathcal{Y}_m\}$.

\subsection{Reward}
Within the framework of RL, the reward is characterized as a scalar signal, indicating the effectiveness of an agent's performance in executing a specific task.
This reward signal functions as a form of environmental feedback, provided to the agent following each action undertaken. The overarching goal of the agent is to establish a policy that optimizes the cumulative sum of rewards over a specified temporal span.
The reward at each environment step that mobile device $m$ can obtain can be denoted as $R=-C(\pmb{\phi}[r], \mathcal{Y}[r])$.
The reward $R$ is primarily related to the state of the mobile device and the chosen action.

\subsection{Model Setting}
The protein folding problem in the field of bioinformatics is a substantial scientific computation task that entails simulating the transition of a protein chain into its three-dimensional formation.
The duration of this process can extend beyond ten hours, or even span several days when using single-machine serial computation--a time slot that many researchers find untenable.
In light of this issue, Robert {\em et al.} \cite{robert2021resource}  have suggested that quantum computers might provide a more expeditious and efficient resolution to this problem compared to classical computing methods.
Thus, in this paper, we generate synthetic data of protein folding simulations based on the following assumptions: the number of amino acids per protein is randomly chosen from the range of $[30, 90]$.
The all-atom molecular dynamics simulation time is randomly assigned within $[100\mu\text{s}, 1\text{ms}]$, and the datasize of each mobile device is set as $[160,320] \times 10^2$ Mb.


\section{Experimental Evaluation}
\subsection{Experimental Settings}
In this section, we assign practical values to each of the parameters mentioned previously, based on real-world considerations. Specifically, we assess and compare the performance of the proposed two-stage algorithm, which melds Lyapunov optimization techniques with DRL, within a hybrid MEQC system. This system spans an area of $\pi(50\times50)\text{~m}^2$, centered around the cloud, and accommodates 15 mobile IoT devices.
The performance evaluation and comparison under varying conditions are reserved on the assumption that the number of physical qubits purchased by each mobile device from the remote server is randomly within the range of $[1000, 5000]$, remaining constant across all time slots. The communication channel model adheres to the sub-band scheme described in \cite{cheng2020vehicular,wang2021general}, with each sub-band allocated a bandwidth of $0.1\text{~GHz}$ to facilitate efficient data transmission.
To simulate heterogeneity among mobile IoT devices, we introduce a randomized selection process for the transmission power of devices within the range of $[0.01, 0.2]\text{~dBm}$.
Similarly, to ascertain the local computation capability of each device, we randomly sample from a predefined set of frequencies $\{1, 2, 3\}\text{~GHz}$. In each time slot, the task volume for each device is randomly generated within the range of $[160, 320]\text{~MB}$. This step-by-step approach ensures a higher degree of coherence and consistency in our analysis.
The processing time $\tau_1, \tau_2, \tau_M$, energy consumption for each physical gate $P_1, P_2, P_M, P_Q$, and the number of gates $N_1, N_2, N_M$ follow the guidelines provided in \cite{ngoenriang2022optimal}.
The hyper-parameters utilized in the DDPG algorithm are as follows.
The hidden layers are constructed using three fully-connected layers consisting of $512$ units. The discounting factor is set to $0.913$ and the learning rate is set to $0.001$. The experiments are conducted using Python 3.8, PyTorch 1.12.1, and CUDA 12.0.

In our research, the primary objective revolves around the development of cost-efficient task offloading strategies for mobile devices operating within hybrid MEQC systems. The decision-making for such policy is heavily influenced by two fundamental factors: the first being the impact originating from the locations of mobile devices, and the second being the consequences of inherent computing capacities, as exhibited by entities such as mobile devices, CPUs, and GPUs.
In order to simulate a realistically representative mobility scenario, we incorporate the GMMM, as demonstrated in \eqref{eq:12} and \eqref{eq:13}, as detailed in the studies by Gao {\em et al.} \cite{gao2019dynamicsocial} and Camp {\em et al.} \cite{camp2002survey}.
Subsequently, we provide a detailed description of the process.
The speed and direction of mobile devices in the GMMM{\footnote{Here, we consider slow mobility scenarios \cite{gao2019dynamicsocial,camp2002survey}, which are basically $3\sim 5\text{~m}/\text{s}$. So, the Doppler frequency shift falls outside the scope of this paper's considerations.}} is given by
\begin{align}\label{eq:12}
\hspace{-0.5em}\begin{cases}
        v_m[r]=\delta v_m[r-1]+(1-\delta)\Bar{v}_m+(1-\delta^2)\tilde{v}_m[r],\\
        \eta_m[r]=\delta \eta_m[r-1]+(1-\delta)\Bar{\eta}_m+(1-\delta^2)\tilde{\eta}_m[r].
    \end{cases}
\end{align}
Here, $0 \le \delta \le 1$ serves as the tuning parameter utilized to adjust the degree of mobility randomness, while $v_m[r]$ and $\eta_m[r]$ represent the speed and direction, respectively, of mobile device $m$ at the time slot $r$.
These values can be derived from the preceding time slot, $r-1$.
The parameters $\Bar{v}_m$ and $\Bar{\eta}_m$ stand for the mean values of velocity and direction respectively as time slot $t$ approaches infinity, with these quantities remaining constant for each individual mobile device $m$.
Conversely, $\tilde{v}_m[r]$ and $\tilde{\eta}_m[r]$ are recognized as random variables that adhere to a standard normal distribution.
In the initial state, all mobile devices are randomly dispersed within the coverage area of the base station.
The parameters of velocity $v_m[r]$ and direction $\eta_m[r]$ pertaining to mobile device $m$ at time slot $r$ exhibit dependence on the position and direction from the preceding time slot, $r-1$.
Upon the mobile device nearing the point of exiting the base station's coverage, an imposed modification occurs in both its immediate and average directional trajectories \cite{camp2002survey}.
The position of mobile device $m$ at time slot $r$ can be deduced from
\begin{align}\label{eq:13}
\begin{cases}
x_m[r]=x_m[r-1]+v_m[r-1]\cos{(\eta_m[r-1])},\\
y_m[r]=y_m[r-1]+v_m[r-1]\sin{(\eta_m[r-1])}.
\end{cases}
\end{align}
where $(x_m[r-1],y_m[r-1])$ and $(x_m[r],y_m[r])$ denote the horizontal and vertical coordinates of mobile device $m$ at time slot $r-1$ and $r$.

\begin{figure*}[!t]
\centering
\subfigure[]{
\begin{minipage}[t]{0.3\linewidth}
\includegraphics[width=2.5in]{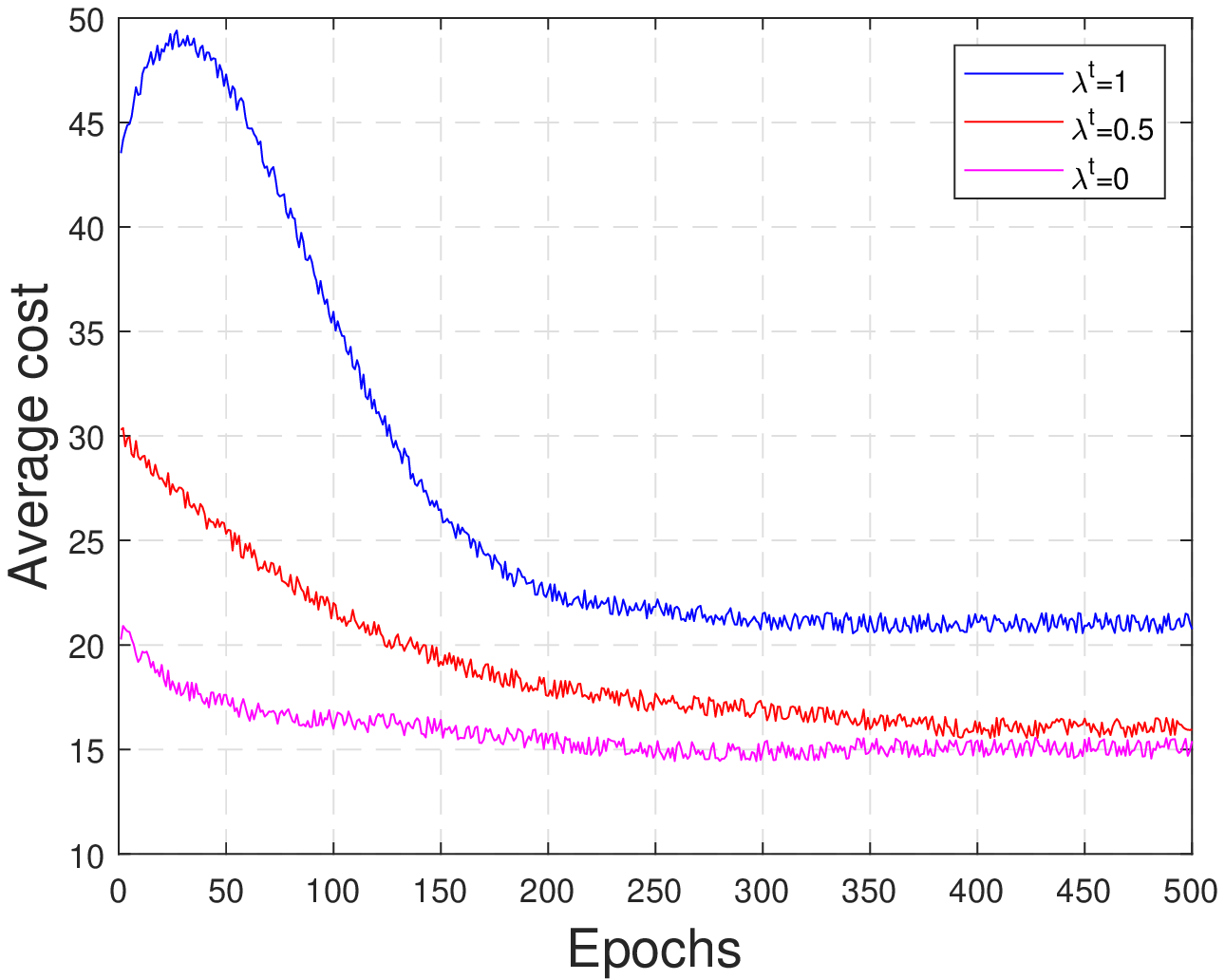}
\end{minipage}
\label{fig:3-1}
}
\subfigure[]{
\begin{minipage}[t]{0.3\linewidth}
\includegraphics[width=2.5in]{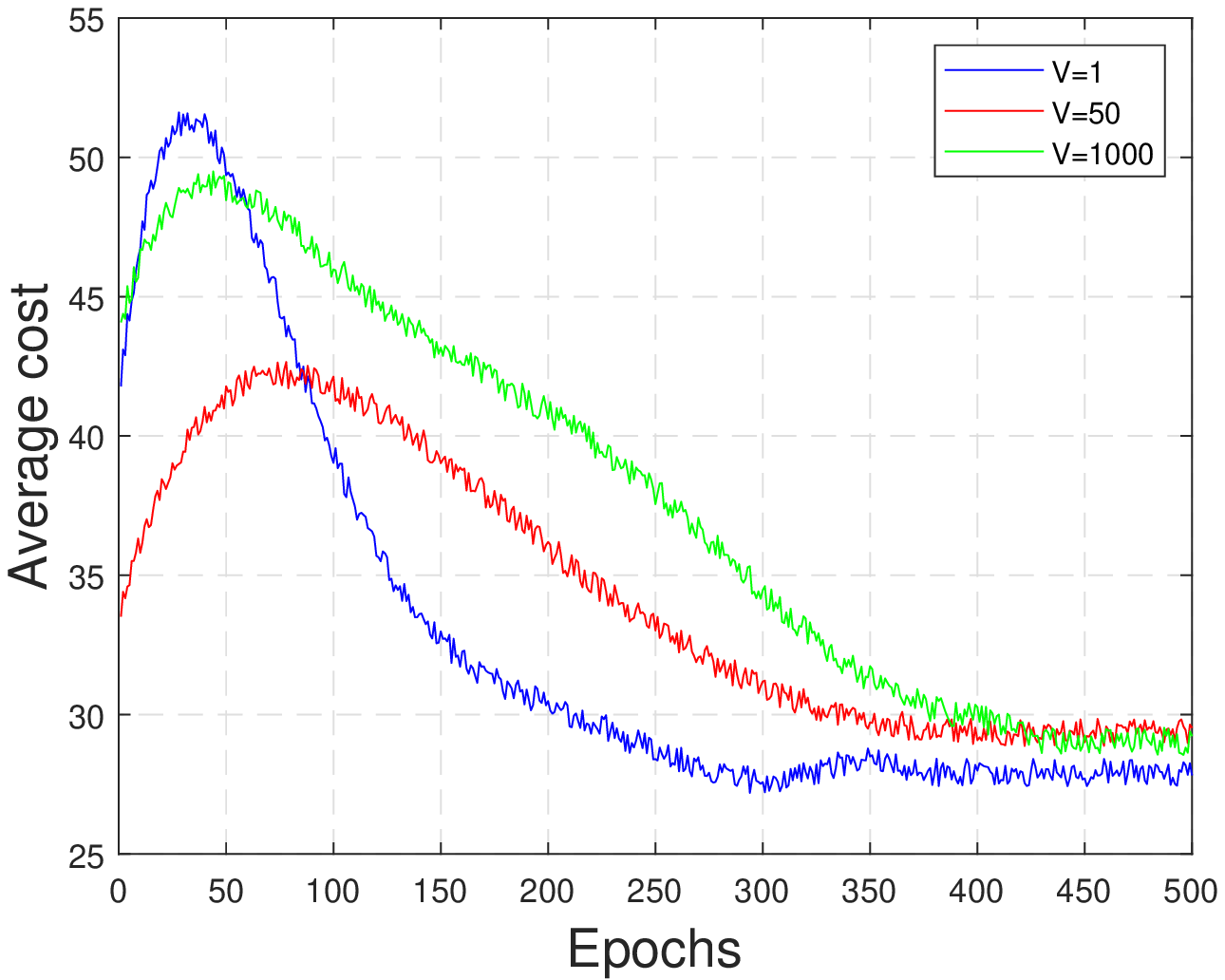}
\end{minipage}
\label{fig:3-2}
}
\subfigure[]{
\begin{minipage}[t]{0.3\linewidth}
\includegraphics[width=2.5in]{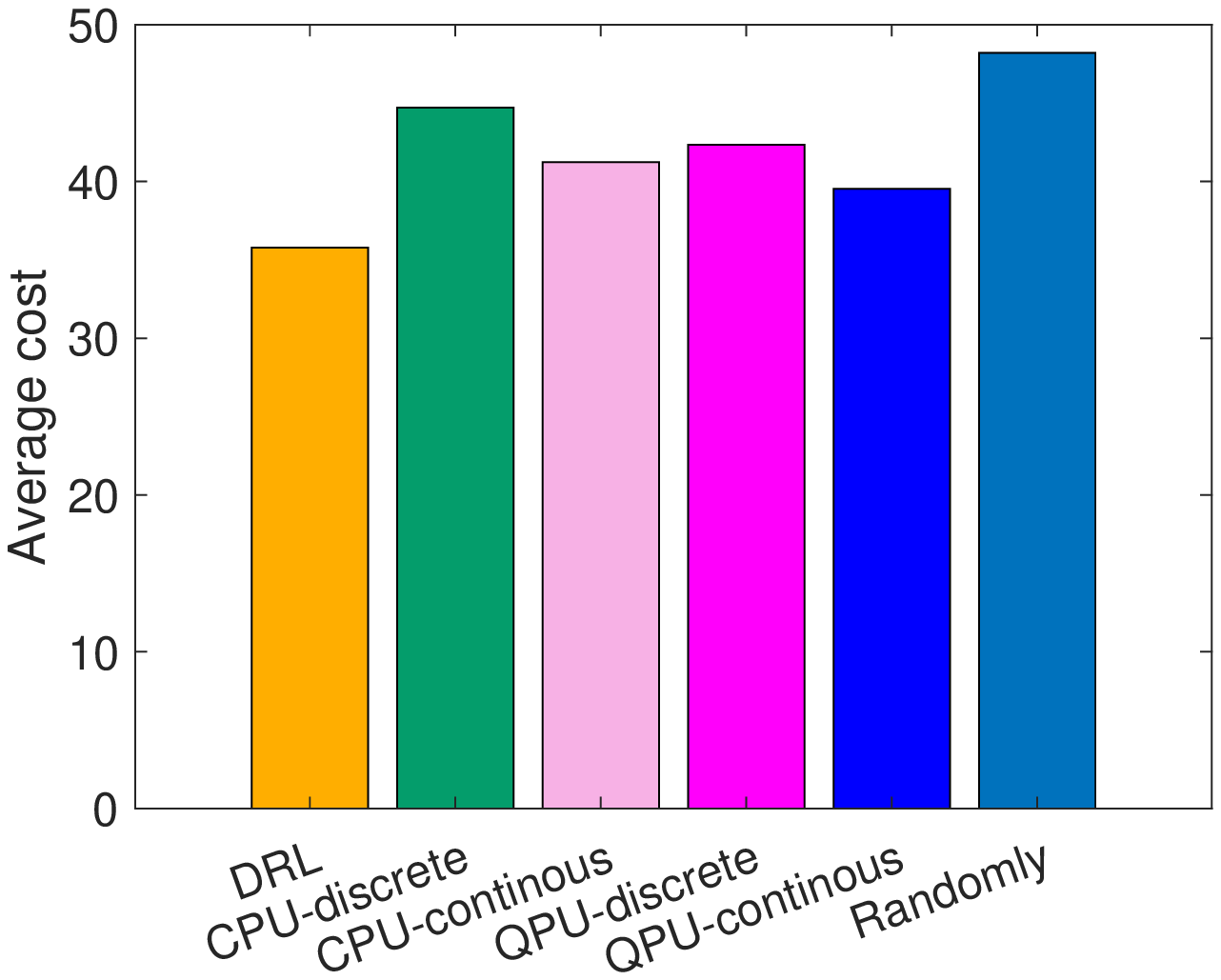}
\end{minipage}
\label{fig:3-3}
}
\caption{(a)Average cost versus epochs in the static scenario,  (b)Average cost versus epochs under different V, and (c)Performance comparison.}
\end{figure*}
\subsection{Effectiveness Verification of the Proposed Algorithm}
Figure \ref{fig:3-1} presents the performance of the DRL-based algorithm with different values of $\lambda_m^t$ within a static context.
From the results, we observe that the time-average cost in the early stage is high due to the randomness of policy in the initial state.
However, in correspondence with an increase in the number of DRL iterations, a steady downward trend is observed in the curves associated with $\lambda^t\in\{1, 0\}$.
The complexity of the objective increases when $\lambda^t=0.5$, thereby the convergence is slowest, taking around $350$ epochs to reach a steady state.
This happens because the algorithm must contemplate a combination of latency and energy consumption optimization, in conjunction with a long-term constraint concerning the proportion of offloading to remote servers.
Devices that are sensitive to energy considerations demonstrate a more rapid convergence compared to their latency-sensitive counterparts, achieving the state of convergence within approximately $200$ iterations.
In contrast, those devices that are delay-sensitive require an estimated $250$ iterations to attain a similar state of convergence.
The results of Fig. \ref{fig:3-1}  suggest that latency presents a more significant cost impact compared to energy consumption.
Furthermore, it's inferred that less-than-optimal network conditions can detrimentally influence user experience, particularly as the data size escalates.

We identify the characteristics of the Lyapunov-based outer loop algorithm, specifically with respect to the epochs (time slots) under different values of the balance parameter $V.$
To this end, the time-average cost performance of the DRL-based Lyapunov algorithm is evaluated by varying $r$ from $0$ to $500$.
Fig. \ref{fig:3-2} shows the impact of the balance parameter $V$ on the time-average cost for the DRL-based Lyapunov algorithm.
As expected,  the time-average cost is asymptotically stable as it evolves over time.
\begin{figure}[!t]
	\centering
	\includegraphics[width=3.5in]{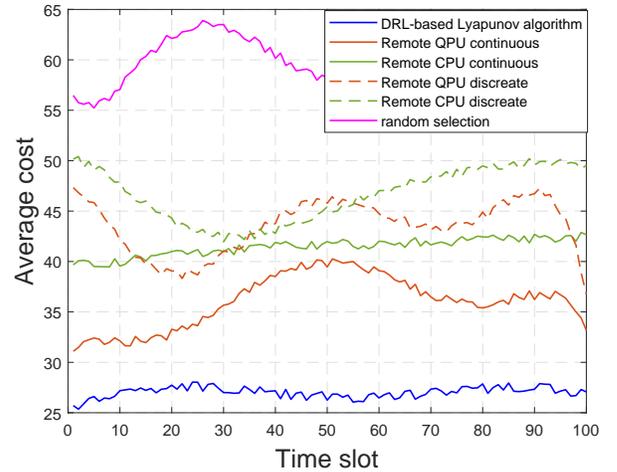}
	\caption{Average cost performance comparison when the mobility of devices follows the GMMM.}
	\label{fig:5}
\end{figure}
\begin{figure*}[!t]
\centering
\subfigure[]{
\begin{minipage}[t]{0.3\linewidth}
\includegraphics[width=2.3in]{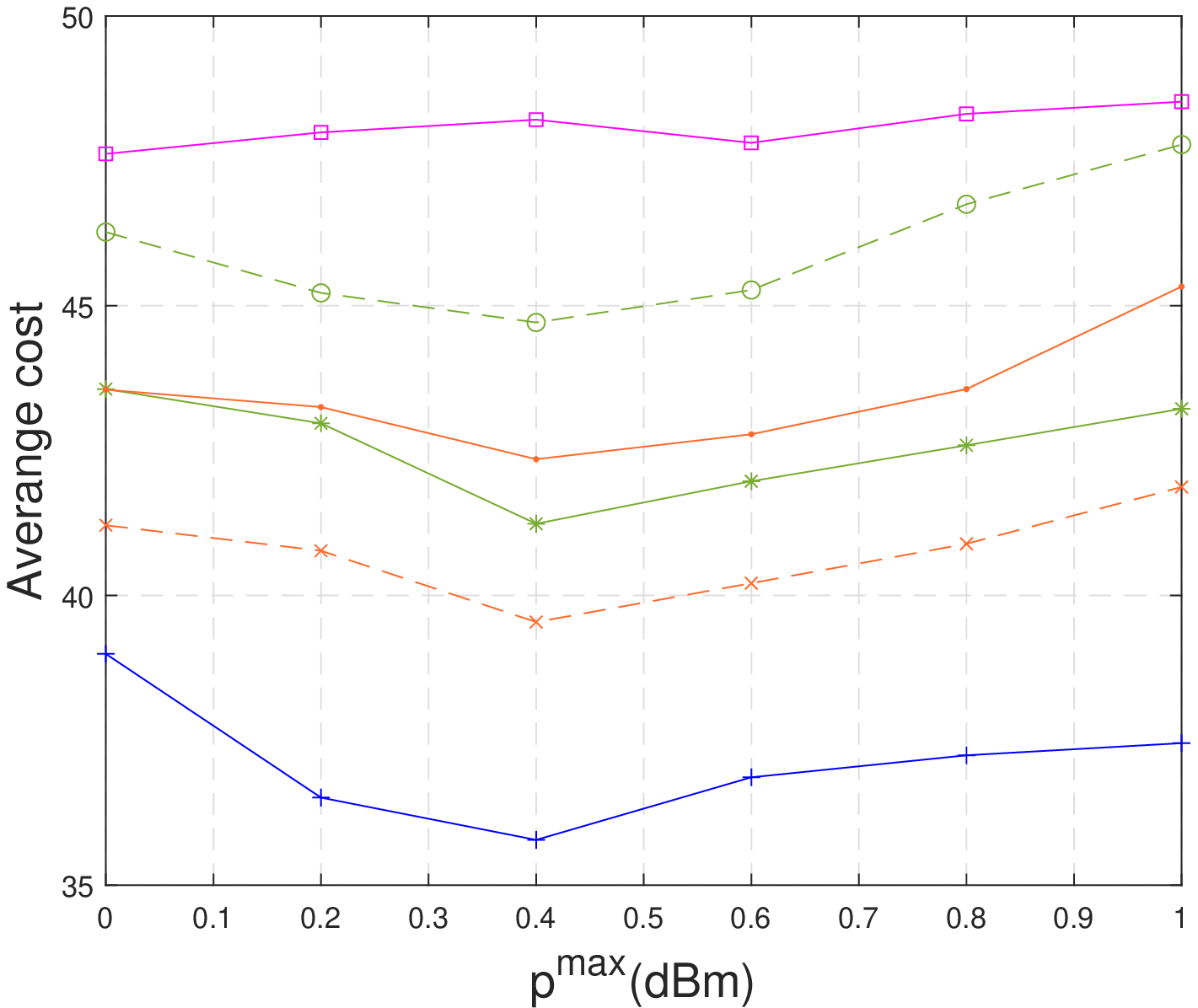}
\end{minipage}
\label{fig:2-1}
}
\subfigure[]{
\begin{minipage}[t]{0.3\linewidth}
\includegraphics[width=2.3in]{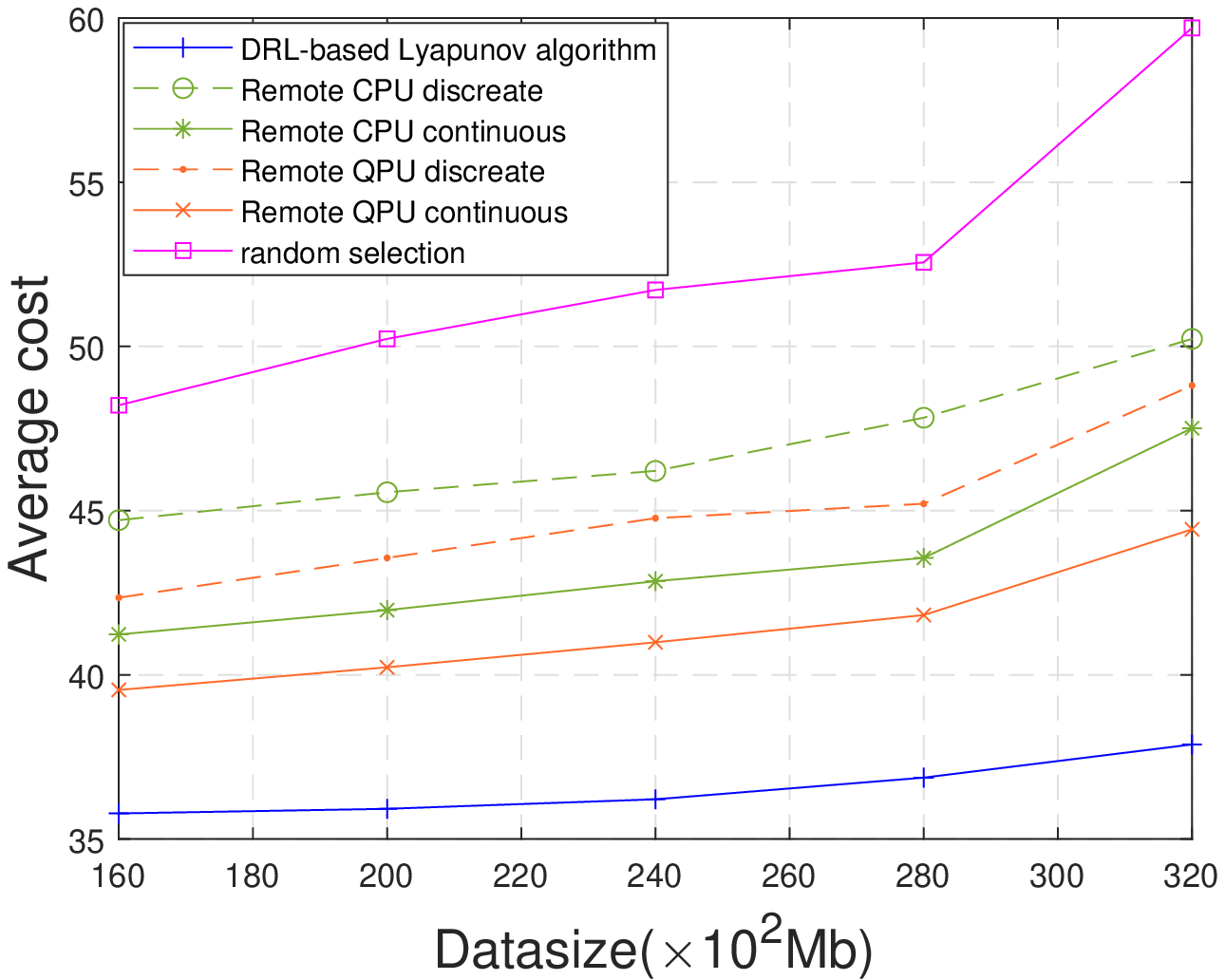}
\end{minipage}
\label{fig:2-2}
}
\subfigure[]{
\begin{minipage}[t]{0.3\linewidth}
\includegraphics[width=2.3in]{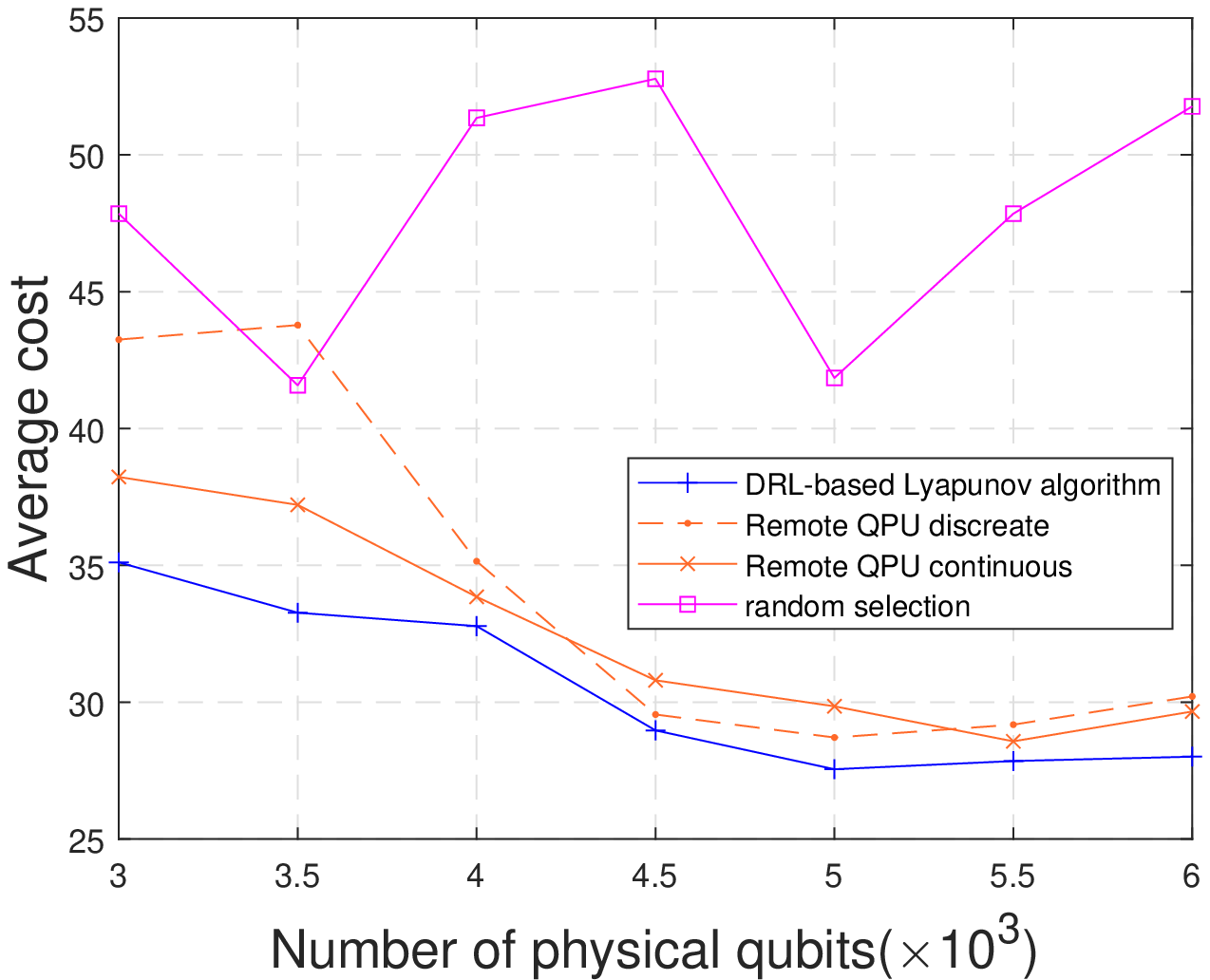}
\end{minipage}
\label{fig:2-3}
}
\caption{ (a) Average cost under different $p^{\max}$, (b) Average cost under different datasize, and (c) Average cost under different number of qubits.}
\label{fig:2}
\end{figure*}

In Fig. \ref{fig:3-3}, we compare the time-average cost of the proposed algorithm against those derived from the following methods. i) CPU-discrete (respectively, CPU-continous): the mobile device can solely access the CPU resources of the remote server, and the device is left with the decision to either offload the entirety (respectively, the fraction) of the task to the remote CPU or not.
ii) QPU-discrete (respectively, QPU-continous): the mobile device can solely access the QPU resources of the remote server, and the device is left with the decision to either offload the entirety (respectively, the fraction) of the task to the remote QPU or not.
iii) Randomly: the mobile device selects the task offloading modes randomly.
As shown in Fig. \ref{fig:3-3}, the DRL-based algorithm for the hybrid MEQC system outperforms other baseline methods.
In line with expectations, the method of randomly selecting the computation offloading mode results in the highest time-average cost, which is largely attributed to the heterogeneity of the mobile devices in aspects such as datasize, computing capabilities, and locations.
In addition, the methods coupled with continuous variable optimization always outperforms discrete variable optimization.
Indeed, facilitating tasks to be offloaded in fractions amplifies the likelihood of realizing cost efficiencies.

Figure \ref{fig:5} depicts the time-average cost performance against time slots for the methods discussed.
The DRL-based Lyapunov algorithm, designed for the hybrid MEQC system, consistently yields the lowest average cost even amidst these dynamic situations.
Conversely, the cost for the method predicated on randomly selecting task offloading modes is markedly higher, exhibiting
substantial fluctuations across various time slots, predominantly due to its disregard of channel conditions.
Further, the method relying solely on cloud servers outfitted  with QPUs demonstrates a lack of robustness, resulting in significant cost oscillations, a consequence of the inherent mobility of devices.

\subsection{Performance Demonstration of the Proposed Method}
\subsubsection{Performance Comparison Versus $p^{\max}$}
Fig. \ref{fig:2-1} plots the time-average cost performance of the aforementioned methods with $p^{\max}$ varied from $0\text{~dBm}$ to $1\text{~dBm}$.
From the results, we observe that the time-average cost of the proposed DRL-based Lyapunov algorithm first decreases and then progressively stabilizes as the maximum transmit power, $p^{\max}$, increases, and outperforms other methods.
It is clear from Fig. \ref{fig:2-1} that an appropriate choice of the transmit budget $p^{\max}$ (i.e.,  $p^{\max}=0.4\text{~dBm}$) can lead the lowest time-average cost.
Furthermore, it's worth noting that when $p^{\max}$ falls within the range of $[0.8, 1]\text{~dBm}$, the average cost associated with the proposed DRL-based Lyapunov algorithm remains relatively stable.
This phenomenon can be attributed to the algorithm's strategic response to increased communication costs, selecting to curtail the volume of tasks offloaded to cloud servers to keep the total cost low.
The random selection system, depicted by the pink line, remains largely unaffected by changes in $p^{\max}$.
However, it consistently yields the highest cost across all $p^{\max}$ values.

\subsubsection{Performance Comparison Versus Datasize}
Fig. \ref{fig:2-2} presents the time-average cost performance of the aforementioned methods with varied the datasize from $160\times 10^2\text{~Mb}$ to $320\times 10^2\text{~Mb}$.
As can be observed from the figure, the time-average cost of all methods increases as datasize of the mobile devices goes up, since the data size is always detrimental of the unilateral system efficiency.
Furthermore, it is noted that our proposed DRL-based Lyapunov algorithm in the hybrid MEQC system outperforms other methods.

\subsubsection{Performance Comparison Versus The Number of Qubits}
Fig. \ref{fig:2-3} shows the influence of the number of qubits in the time-average cost for all the mentioned methods.
From this figure, we observe that the average cost of the DRL-based Lyapunov algorithm is not significantly impacted by an escalation in the number of physical qubits.
The rationale lies in the fact that the quantity of qubits present at the edge server stipulates the probability of mobile users gaining a quantum advantage.
As the count of physical qubits increases, there is a tendency for mobile users to offload a greater fraction of tasks to the QPU, consequently leading to a reduction in costs when the network conditions are conducive.
\begin{figure}[!t]
	\centering
	\includegraphics[width=3.5in]{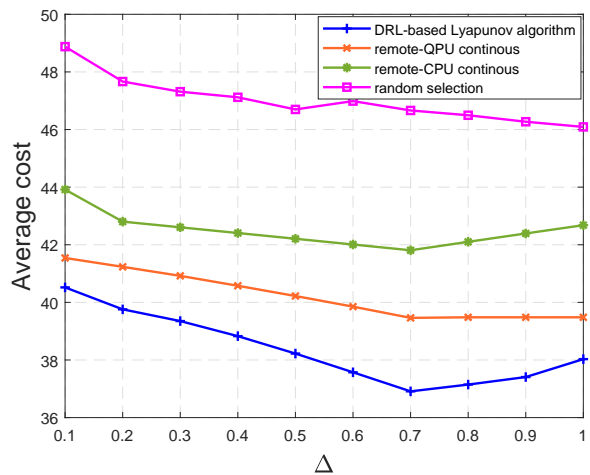}
	\caption{Average cost versus $\Delta$.}
	\label{fig:6}
\end{figure}

\subsubsection{Impact of $\Delta$}
Fig. \ref{fig:6} depicts the time-average cost of the referenced methods against the expected rate of mobile devices opting for remote offloading, denoted as $\Delta$.
For any given values of $p^{\max}$ and datasize as well as the number of qubits, there is an optimal choice of $\Delta$ (i.e., $0.7$), which leads to the minimum time-average cost.
It is observable that for the expected rate values situated to the left of the optimal point, the time-average cost of the proposed DRL-based Lyapunov algorithm decreases as the value of the value of $\Delta$ increases.
Conversely, for expected rate values to the right of the optimal point, the time-average cost increases as the value of $\Delta$ rises.
From the results, it is additionally evident that for other mentioned methods, the time-average cost undergoes negligible change when modifying the value of $\Delta$.

\section{Conclusions}
In this paper, we explore the integration of hybrid MEQC for sustainable task offloading scheduling in mobile networks. Our main focus is on developing strategies that are both sustainable and cost-efficient for selecting task offloading mode and allocating computation and communication resources. This involves addressing the challenges of the time-coupled  optimization problem.
To address these challenges, we propose an effective and computationally affordable approach that combines Lyapunov optimization techniques with DRL algorithms. Specifically, we implement a DRL algorithm to make partial-task offloading decisions and test it in a realistic network setting.
The experimental results from our numerical simulations indicate that our proposed DRL-based Lyapunov algorithm outperforms the baseline methods in terms of cost-effectiveness. This suggests that our approach holds promise for achieving sustainable and efficient task offloading within mobile networks.

\begin{appendices}
\section{Proof of Theorem \ref{theorem:1}}
If action $\phi_m[r]$ is a $C\text{-additive}$ approximation, we have
\begin{align}\label{add:1}
    \mathbb{E}&\left\{-VC(\pmb\phi[r],Y[r])+\sum_{m=1}^Mz_m[r](\Delta {\boldsymbol 1}_{[\phi_m[r]=0]}-\phi_m[r])|\mathcal{Z}[r]\right\} \notag \\
    &\le C+\mathbb{E}\left \{ -VC({\pmb\phi}^{*}[r],Y^*[r])+
    \sum\nolimits_{m=1}^Mz_m[r]\right.\notag \\
    &\left.~~\times(\Delta{\boldsymbol 1}_{[\phi_m[r]=0]} -\phi_m[r])\right\},
\end{align}
where ${\pmb\phi}^*[r],Y^*[r]$ is an i.i.d. algorithm.

Applying the aforementioned Eqs. \eqref{eq:14} and \eqref{eq:15} to \eqref{add:1}, we can arrive at the following inequality
\begin{align}\label{eq:16}
    \mathbb{E}&\left\{ \left ( \Delta L[r]-VC({\pmb\phi}[r],Y[r])\right )|\mathcal{Z}[r]\right\}\notag \\
    \le& B+C+\mathbb{E}\left \{ -VC({\pmb\phi}^*[r],Y^*[r])
    +\sum\nolimits_{m=1}^Mz_m[r]\right.\notag \\
    &\left.\times(\Delta {\boldsymbol 1}_{[\phi_m[r]=0]}-\phi_m[r]) \right\}.
\end{align}\label{eq:17}
Moreover, fixing any $\sigma<0$, the i.i.d. algorithm ${\pmb\phi}^*[r],Y^*[r]$ satisfies
\begin{align}\label{eq:17}
    \mathbb{E}\left\{ C({\pmb\phi}^*[r],Y^*[r])\right \} \ge C^{\text{opt}} + \sigma.
\end{align}
By combining \eqref{eq:16} and \eqref{eq:17} we observe that
\begin{align}\label{eq:18}
    \mathbb{E}\left\{ \left(\Delta L[r]-VC({\pmb\phi}[r],Y[r]) \right)|\mathcal{Z}[r]\right\} \le
    B+C-VC^{\text{opt}}.
\end{align}
The above holds for all $\sigma <0$. Taking a limit as $\sigma \to 0$ in the above yields
\begin{align}
    \frac{\mathbb{E}\{L[\tau]\}-\mathbb{E}\{L[0]\}-V\mathbb{E}\left\{\sum_{r=0}^{\tau=1}C({\pmb\phi}[r],Y[r])\right\}}{\tau} \notag \\
    \le B+C-VC^{\text{opt}}.
\end{align}
Thus, the result of Theorem 1 can be derived since the non-negative of quadratic Lyapunov function and $L[0]=0$.
\section{Proof of Theorem \ref{theorem:2}}
Taking expectations of \eqref{eq:18} gives
\begin{align}
    \mathbb{E}\left\{ \Delta L[r]\right\}-V\mathbb{E}\left\{ C[r]\right\} \le B+C-VC^{\text{opt}}.
\end{align}
Using $\mathbb{E}\left\{ L[0]\right\}=0$ and $L[r]=\frac{1}{2}\sum_{m=1}^{M}z_m[r]^2$, summing over $r \in \left\{ 0,1,2,\dots,\tau-1\right\}$ gives
\begin{align}
    \frac{1}{2} \mathbb{E}\left\{\sum\nolimits_{m=1}^{M} z_{m}[\tau]^{2}\right\} \leq& \tau(B+C)+V \mathbb{E}\left\{\sum\nolimits_{r=0}^{\tau-1} C[r]\right\} \notag \\
    &-V C^{\text {opt }} \tau.
\end{align}
Dividing by $\tau^2$ and taking a square root yields
\begin{align}
    \frac{\mathbb{E}\left\{\sum_{m=1}^{M} z_{m}[\tau]^{2}\right\}}{\tau} \leq \sqrt{\frac{2(B+C)+V\left(C[\tau]-C^{\mathrm{opt}}\right)}{\tau}}.
\end{align}
Then we have
\begin{align}
    (\Delta {\boldsymbol 1}_{[\phi_m[r]=0]}-\phi_m[r])\leq \sqrt{\frac{2(B+C)+V\left(C[\tau]-C^{\mathrm{opt}}\right)}{\tau}}.
\end{align}
\end{appendices}



\begin{thebibliography}{10}
\providecommand{\url}[1]{#1}
\csname url@samestyle\endcsname
\providecommand{\newblock}{\relax}
\providecommand{\bibinfo}[2]{#2}
\providecommand{\BIBentrySTDinterwordspacing}{\spaceskip=0pt\relax}
\providecommand{\BIBentryALTinterwordstretchfactor}{4}
\providecommand{\BIBentryALTinterwordspacing}{\spaceskip=\fontdimen2\font plus
\BIBentryALTinterwordstretchfactor\fontdimen3\font minus
  \fontdimen4\font\relax}
\providecommand{\BIBforeignlanguage}[2]{{%
\expandafter\ifx\csname l@#1\endcsname\relax
\typeout{** WARNING: IEEEtran.bst: No hyphenation pattern has been}%
\typeout{** loaded for the language `#1'. Using the pattern for}%
\typeout{** the default language instead.}%
\else
\language=\csname l@#1\endcsname
\fi
#2}}
\providecommand{\BIBdecl}{\relax}
\BIBdecl

\bibitem{kan2018task}
T.-Y. Kan, Y.~Chiang, and H.-Y. Wei, ``Task offloading and resource allocation
  in mobile-edge computing system,'' in \emph{2018 27th wireless and optical
  communication conference (WOCC)}.\hskip 1em plus 0.5em minus 0.4em\relax
  IEEE, 2018, pp. 1--4.

\bibitem{mahenge2022energy}
M.~P.~J. Mahenge, C.~Li, and C.~A. Sanga, ``Energy-efficient task offloading
  strategy in mobile edge computing for resource-intensive mobile
  applications,'' \emph{Digital Communications and Networks}, 2022.

\bibitem{cheng2021research}
J.~Cheng and D.~Guan, ``Research on task-offloading decision mechanism in
  mobile edge computing-based internet of vehicle,'' \emph{EURASIP Journal on
  Wireless Communications and Networking}, vol. 2021, no.~1, pp. 1--14, 2021.

\bibitem{9575181}
R.~Zhang, K.~Xiong, Y.~Lu, B.~Gao, P.~Fan, and K.~B. Letaief, ``Joint
  coordinated beamforming and power splitting ratio optimization in mu-miso
  swipt-enabled hetnets: A multi-agent ddqn-based approach,'' \emph{IEEE
  Journal on Selected Areas in Communications}, vol.~40, no.~2, pp. 677--693,
  2022.

\bibitem{10032267}
R.~Zhang, K.~Xiong, Y.~Lu, P.~Fan, D.~W.~K. Ng, and K.~B. Letaief, ``Energy
  efficiency maximization in ris-assisted swipt networks with rsma: A ppo-based
  approach,'' \emph{IEEE Journal on Selected Areas in Communications}, vol.~41,
  no.~5, pp. 1413--1430, 2023.

\bibitem{mechealth}
\BIBentryALTinterwordspacing
{MEC} health. [Online]. Available:
  \url{https://www.espon.eu/sites/default/files/attachments/Final\%20report.\%202019\%2003\%2025_final\%20version_0.pdf}
\BIBentrySTDinterwordspacing

\bibitem{al2017technologies}
N.~Al-Falahy and O.~Y. Alani, ``Technologies for {5G} networks: Challenges and
  opportunities,'' \emph{It Professional}, vol.~19, no.~1, pp. 12--20, 2017.

\bibitem{pittenger2012introduction}
A.~O. Pittenger, \emph{An introduction to quantum computing algorithms}.\hskip
  1em plus 0.5em minus 0.4em\relax Springer Science \& Business Media, 2012,
  vol.~19.

\bibitem{rietsche2022quantum}
R.~Rietsche, C.~Dremel, S.~Bosch, L.~Steinacker, M.~Meckel, and J.-M.
  Leimeister, ``Quantum computing,'' \emph{Electronic Markets}, pp. 1--12,
  2022.

\bibitem{ngoenriang2022optimal}
N.~Ngoenriang, M.~Xu, S.~Supittayapornpong, D.~Niyato, H.~Yu \emph{et~al.},
  ``Optimal stochastic resource allocation for distributed quantum computing,''
  \emph{arXiv preprint arXiv:2210.02886}, 2022.

\bibitem{cicconetti2022resource}
C.~Cicconetti, M.~Conti, and A.~Passarella, ``Resource allocation in quantum
  networks for distributed quantum computing,'' in \emph{2022 IEEE
  International Conference on Smart Computing (SMARTCOMP)}.\hskip 1em plus
  0.5em minus 0.4em\relax IEEE, 2022, pp. 124--132.

\bibitem{ngoenriang2023dqc2o}
N.~Ngoenriang, M.~Xu, J.~Kang, D.~Niyato, H.~Yu, and X.~S. Shen, ``Dqc2o:
  Distributed quantum computing for collaborative optimization in future
  networks,'' \emph{IEEE Communications Magazine}, 2023.

\bibitem{ajagekar2019quantum}
A.~Ajagekar and F.~You, ``Quantum computing for energy systems optimization:
  Challenges and opportunities,'' \emph{Energy}, vol. 179, pp. 76--89, 2019.

\bibitem{ravi2021quantum}
G.~S. Ravi, K.~N. Smith, P.~Gokhale, and F.~T. Chong, ``Quantum computing in
  the cloud: Analyzing job and machine characteristics,'' in \emph{2021 IEEE
  International Symposium on Workload Characterization (IISWC)}.\hskip 1em plus
  0.5em minus 0.4em\relax IEEE, 2021, pp. 39--50.

\bibitem{ravi2021adaptive}
G.~S. Ravi, K.~N. Smith, P.~Murali, and F.~T. Chong, ``Adaptive job and
  resource management for the growing quantum cloud,'' in \emph{2021 IEEE
  International Conference on Quantum Computing and Engineering (QCE)}.\hskip
  1em plus 0.5em minus 0.4em\relax IEEE, 2021, pp. 301--312.

\bibitem{nakaiqubit}
M.~Nakai, ``Qubit allocation for distributed quantum computing.''

\bibitem{boyd2004convex}
S.~Boyd, S.~P. Boyd, and L.~Vandenberghe, \emph{Convex optimization}.\hskip 1em
  plus 0.5em minus 0.4em\relax Cambridge university press, 2004.

\bibitem{chen2022dynamic}
Y.~Chen, F.~Zhao, Y.~Lu, and X.~Chen, ``Dynamic task offloading for mobile edge
  computing with hybrid energy supply,'' \emph{Tsinghua Science and
  Technology}, vol.~28, no.~3, pp. 421--432, 2022.

\bibitem{tang2020deep}
M.~Tang and V.~W. Wong, ``Deep reinforcement learning for task offloading in
  mobile edge computing systems,'' \emph{IEEE Transactions on Mobile
  Computing}, vol.~21, no.~6, pp. 1985--1997, 2020.

\bibitem{li2020genetic}
Z.~Li and Q.~Zhu, ``Genetic algorithm-based optimization of offloading and
  resource allocation in mobile-edge computing,'' \emph{Information}, vol.~11,
  no.~2, p.~83, 2020.

\bibitem{you2021efficient}
Q.~You and B.~Tang, ``Efficient task offloading using particle swarm
  optimization algorithm in edge computing for industrial internet of things,''
  \emph{Journal of Cloud Computing}, vol.~10, pp. 1--11, 2021.

\bibitem{gao2019dynamic}
Y.~Gao, W.~Tang, M.~Wu, P.~Yang, and L.~Dan, ``Dynamic social-aware computation
  offloading for low-latency communications in {IoT},'' \emph{IEEE Internet of
  Things Journal}, vol.~6, no.~5, pp. 7864--7877, 2019.

\bibitem{gao2022multi}
Y.~Gao, Z.~Ye, H.~Yu, Z.~Xiong, Y.~Xiao, and D.~Niyato, ``Multi-resource
  allocation for on-device distributed federated learning systems,'' in
  \emph{GLOBECOM 2022-2022 IEEE Global Communications Conference}.\hskip 1em
  plus 0.5em minus 0.4em\relax IEEE, 2022, pp. 160--165.

\bibitem{steane1998quantum}
A.~Steane, ``Quantum computing,'' \emph{Reports on Progress in Physics},
  vol.~61, no.~2, p. 117, 1998.

\bibitem{mystakidis2022metaverse}
S.~Mystakidis, ``Metaverse,'' \emph{Encyclopedia}, vol.~2, no.~1, pp. 486--497,
  2022.

\bibitem{ma2022hybrid}
L.~Ma and L.~Ding, ``Hybrid quantum edge computing network,'' in \emph{Proc. of
  {SPIE} {Vol}}, vol. 12238, 2022, pp. 122\,380F--1.

\bibitem{zaman2023quantum}
F.~Zaman, A.~Farooq, M.~A. Ullah, H.~Jung, H.~Shin, and M.~Z. Win, ``Quantum
  machine intelligence for {6G} {URLLC},'' \emph{IEEE Wireless Communications},
  vol.~30, no.~2, pp. 22--30, 2023.

\bibitem{xu2022learning}
M.~Xu, D.~Niyato, J.~Kang, Z.~Xiong, and M.~Chen, ``Learning-based sustainable
  multi-user computation offloading for mobile edge-quantum computing,''
  \emph{arXiv preprint arXiv:2211.06681}, 2022.

\bibitem{wang2021resource}
D.~Wang, B.~Song, P.~Lin, F.~R. Yu, X.~Du, and M.~Guizani, ``Resource
  management for edge intelligence ({EI})-assisted {IoV} using quantum-inspired
  reinforcement learning,'' \emph{IEEE Internet of Things Journal}, vol.~9,
  no.~14, pp. 12\,588--12\,600, 2021.

\bibitem{chamberlain2008visions}
R.~D. Chamberlain, J.~M. Lancaster, and R.~K. Cytron, ``Visions for application
  development on hybrid computing systems,'' \emph{Parallel Computing},
  vol.~34, no. 4-5, pp. 201--216, 2008.

\bibitem{fellous2022optimizing}
M.~Fellous-Asiani, J.~H. Chai, Y.~Thonnart, H.~K. Ng, R.~S. Whitney, and
  A.~Auff{\`e}ves, ``Optimizing resource efficiencies for scalable full-stack
  quantum computers,'' \emph{arXiv preprint arXiv:2209.05469}, 2022.

\bibitem{krinner2022realizing}
S.~Krinner, N.~Lacroix, A.~Remm, A.~Di~Paolo, E.~Genois, C.~Leroux,
  C.~Hellings, S.~Lazar, F.~Swiadek, J.~Herrmann \emph{et~al.}, ``Realizing
  repeated quantum error correction in a distance-three surface code,''
  \emph{Nature}, vol. 605, no. 7911, pp. 669--674, 2022.

\bibitem{neely2012dynamic}
M.~J. Neely, ``Dynamic optimization and learning for renewal systems,''
  \emph{IEEE Transactions on Automatic Control}, vol.~58, no.~1, pp. 32--46,
  2012.

\bibitem{robert2021resource}
A.~Robert, P.~K. Barkoutsos, S.~Woerner, and I.~Tavernelli,
  ``Resource-efficient quantum algorithm for protein folding,'' \emph{npj
  Quantum Information}, vol.~7, no.~1, pp. 1--5, 2021.

\bibitem{cheng2020vehicular}
X.~Cheng, Z.~Huang, and S.~Chen, ``Vehicular communication channel measurement,
  modelling, and application for beyond {5G} and {6G},'' \emph{IET
  Communications}, vol.~14, no.~19, pp. 3303--3311, 2020.

\bibitem{wang2021general}
J.~Wang, C.-X. Wang, J.~Huang, H.~Wang, and X.~Gao, ``A general {3D}
  space-time-frequency non-stationary {THz} channel model for {6G}
  ultra-massive {MIMO} wireless communication systems,'' \emph{IEEE Journal on
  Selected Areas in Communications}, vol.~39, no.~6, pp. 1576--1589, 2021.

\bibitem{gao2019dynamicsocial}
Y.~Gao, Y.~Xiao, M.~Wu, M.~Xiao, and J.~Shao, ``Dynamic social-aware peer
  selection for cooperative relay management with {D2D} communications,''
  \emph{IEEE Transactions on Communications}, vol.~67, no.~5, pp. 3124--3139,
  2019.

\bibitem{camp2002survey}
T.~Camp, J.~Boleng, and V.~Davies, ``A survey of mobility models for ad hoc
  network research,'' \emph{Wireless communications and mobile computing},
  vol.~2, no.~5, pp. 483--502, 2002.

\end{thebibliography}


\end{document}